\documentclass[conference]{IEEEtran}
\IEEEoverridecommandlockouts
\usepackage{cite}
\usepackage{amsmath,amssymb,amsfonts}
\usepackage{algorithmic}
\usepackage{graphicx}
\usepackage{textcomp}
\usepackage{xcolor}
\usepackage{bbm}
\def\BibTeX{{\rm B\kern-.05em{\sc i\kern-.025em b}\kern-.08em
    T\kern-0.1667em\lower.7ex\hbox{E}\kern-0.125emX}}
\usepackage{hyperref}    
    \usepackage{amsthm}
    \newtheorem{theorem}{Theorem}[section]

\newcommand{\Id}{\mathbbm{1}} 
\newcommand{\dd}{\ \mathrm{d}}  
\newcommand{\E}{\mathbb{E}} 
\newcommand{\I}{\mathbb{I}} 
\newcommand{\Cov}{\mathbb{C}\mathrm{ov}}
\newcommand{\Prob}{\mathbb{P}} 
\newcommand{\R}{\mathbb{R}} 
\newcommand{\N}{\mathbb{N}} 
 
\newcommand{\sgn}{\mathrm{sgn}} 

\newmuskip\pFqmuskip

\newcommand*\pFq[6][8]{%
  \begingroup 
  \pFqmuskip=#1mu\relax
  \mathchardef\normalcomma=\mathcode`,
  \mathcode`\,=\string"8000
  \begingroup\lccode`\~=`\,
  \lowercase{\endgroup\let~}\pFqcomma
  F^{#2}_{#3}{\left[\genfrac..{0pt}{}{#4}{#5};#6\right]}%
  \endgroup
}
\newcommand{\pFqcomma}{{\normalcomma}\mskip\pFqmuskip}

\newcommand{\msg}{X}    
\newcommand{\out}{Y}    
\newcommand{\rbd}{r_0}  
\newcommand{\state}{x}  
\newcommand{\gain}{c_3}   
\newcommand{\condmean}{Z} 

\newcommand{\statmean}{m}
\newcommand{\z}{z}
\newcommand{\sojourn}{\sigma} 
\newcommand{\constr}[1]{(C{#1})} 

\begin{document}

\title{Poisson channel with binary Markov input and average sojourn time constraint\\
\thanks{The work was supported by the European Research Council (ERC) within the Consolidator Grant CONSYN (grant agreement no. 773196).\newline \newline
This article was accepted for publication by IEEE, ISIT 2020. \newline
doi: \href{https://ieeexplore.ieee.org/document/9174360}{10.1109/ISIT44484.2020.9174360}
\textcopyright 2020 IEEE}
}

\author{\IEEEauthorblockN{Mark Sinzger, Maximilian Gehri, Heinz Koeppl}
\IEEEauthorblockA{\textit{Dept. of Electrical Engineering, Centre for Synthetic Biology},
\textit{Technische Universit{\"a}t Darmstadt}\\
Darmstadt, Germany \\
Email: \{mark.sinzger,maximilian.gehri,heinz.koeppl\}@bcs.tu-darmstadt.de}

}

\maketitle

\begin{abstract}

A minimal model for gene expression, consisting of a switchable promoter together with the resulting messenger RNA, is equivalent to a Poisson channel with a binary Markovian input process. Determining its capacity is an optimization problem with respect to two parameters: the average sojourn times of the promoter's active (ON) and inactive (OFF) state. An expression for the mutual information is found by solving the associated filtering problem analytically on the level of distributions. For fixed peak power, three bandwidth-like constraints are imposed by lower-bounding (i) the average sojourn times (ii) the autocorrelation time and (iii) the average time until a transition.
OFF-favoring optima are found for all three constraints, as commonly encountered for the Poisson channel. In addition, constraint (i) exhibits a region that favors the ON state, and (iii) shows ON-favoring local optima. 
\end{abstract}

\begin{IEEEkeywords}
Poisson channel, gene expression, binary Markov, average sojourn time, filtering, bandwidth constraint
\end{IEEEkeywords}

\section{Introduction}

There is mounting evidence that the information encoded in the temporal concentration profiles of biomolecules plays a key role in cellular sensing and decision making \cite{PURVIS2013945, friedrich2019stochastic} 
and helps to overcome biochemical noise \cite{Selimkhanov1370}. The computation of mutual information (MI) between time-varying, biomolecular signals is complex, and analytical solutions so far relied on Gaussian \cite{TenWoldeInOut2009} or steady-state approximations \cite{MuglerWalczakWiggins2009}. Other papers have focused on single time-point transmission, e.g., \cite{suderman2017fundamental,tkavcik2008information}.  Works, that account for the discrete nature of chemical reactions, commonly assume diffusion approximations as inputs \cite{lestas2010fundamental,Nakahira2018} or are based on stochastic simulation \cite{cepedahumerez2018estimating,duso2019path,PashaSolo2012particlefilter}. Restrictions to sub-classes of discrete-state input processes permit analytical bounds on the capacity \cite{shamai1990pam}, often challenging diffusion based results \cite{parag2019signalling}. 

This paper analyzes the minimal gene expression model with a two-state promoter (see e.g. \cite{PeskinGeneExpressionModel} and Fig. 2a in \cite{ZechnerKoeppl2014}) as an analytically tractable example.
Switching stochastically between activated state $x_1$ and inactivated state $x_2$, the promoter is modeled by a stationary random telegraph process $\msg(t)$, i.e., a binary, time-homogeneous, stationary Markov process (BMP). Its state linearly modulates the synthesis rate of messenger RNA molecules $\out(t)$. The decay of mRNA molecules can be ignored from an information theoretic point of view, because birth events are uniquely identified from birth-death trajectories \cite{duso2019path}.
The joint distribution of $(\msg, \out)$, that factorizes in the conditional $\out \vert \msg$ and the input path distribution $\mu_\msg$, is equivalent to the Poisson channel, whose class of input processes $\msg(t)$ is restricted to BMPs.
We distinguish between channels with leakage ($\state_1 > 0$) and without leakage ($\state_1 = 0$), see Fig. \ref{fig:cond_mean}a. By $\msg_{[0,T]}$ we denote the trajectory $\msg(t)_{0 \le t \le T}$ of the time-varying input signal with transmission duration $T$. The sojourn times $\sojourn_{1}$ in $\state_1$ and $\sojourn_{2}$ in $\state_2$ are exponentially distributed with parameters $c_1$ and $c_2$ such that
$\E[\sojourn_{1}] = c_{1}^{-1}$ and $\E[\sojourn_{2}] = c_{2}^{-1}$.
The channel output $\out(t)$ fires at rate $c_3\msg(t)$, where $c_3$ is the channel gain that dictates the time scale of $\out(t)$. We denote the jump times of $\out$ as $t_i$ with $0 < t_1 < \cdots < T$.

The Poisson channel was introduced as a model for direct-detection optical communication systems \cite{MacchiPicinbono}. Classically, peak and average power are constrained. Then among general inputs the class of BMPs achieves capacity, however at the cost of infinite switching rates $c_1, c_2$  \cite{kabanov1978capacity,davis1980capacity}. This physical implausibility motivates bandwidth-like constraints \cite{snyder1983design,shamai_lapidoth1993poisson}. By restricting general signals to binary inputs with lower-bounded sojourn times, \cite{shamai1991capacity} reported a transition from asymmetric to symmetric allocation for the capacity-achieving input, as the lower bound increases.

\subsection{Problem statement and outline}
We consider the Poisson channel with BMPs as input class and investigate the following bandwidth-like constraints:
\begin{itemize}
\item (C1) $0 <c_1 \le r_1, 0<c_2 \le r_2$, i.e., lower-bounding the average sojourn times by $\E[\sojourn_{1}] \geq r_1^{-1}, \E[\sojourn_{2}] \geq r_2^{-1}$. A special case is the homogeneous constraint $\max\{ c_1, c_2 \} \le r_0$, in analogy to \cite{shamai1991capacity}.
    \item (C2) $0 <c_1 + c_2 \le r_0$. Since $\Cov[X(t),X(t+s)] \propto \exp(-(c_1 + c_2) s)$, this bounds the autocorrelation time of $\msg(t)$ from below.
    \item (C3) $\E[\E[\sojourn_{\msg(t)}| \msg(t)]]=c_1^{-1} + c_2^{-1} - 2/(c_1 + c_2) \geq r_0^{-1}$, which lower-bounds the average sojourn time, similarly to (C1), but regardless of the transition type. 
    \end{itemize}

As the biophysical interpretation and value of the rates $c_1,c_2$ strongly depends on the context of the promoter model \cite{ZechnerKoeppl2014,TranscrByNumbers1,suter2011mammalian,Mirny22534} we list some generic motivations. The timescales of the rates may be determined by binding affinity to the promoter, temperature, diffusion \cite{Tkacik2011Review,Bialek2005Physicallimits}, and availability of activating and deactivating constituents, such as transcription factors, polymerases, other enzymes, or ATP. Highly autocorrelated dynamics may be caused by upstream regulation or feedback of downstream elements in a reaction network \cite{lestas2010fundamental,KIM20111167}. 

We consider the path-wise MI $\I(\msg_{[0,T]}, \out_{[0,T]})$ \cite{liptser2001statistics} and the information rate, defined as
$$\bar{\I}(\msg, \out):= \lim_{T \to \infty} \frac 1 T \I(\msg_{[0,T]}, \out_{[0,T]}).$$
Optimizing (i) the MI and (ii) the information rate with respect to all admissible 
input path distributions $\mu_\msg$ yields (i) the capacity $C_T$ and (ii) the information rate capacity $C$. For fixed OFF and ON states $\state_1, \state_2$, the distribution $\mu_\msg$ is parametrized solely by the system parameters $c_1, c_2 \geq 0$. Fixing $c_3$, the respective capacity-achieving path distributions are characterized by the following constrained optimization problems
$$ \mathrm{(i)} \; \max_{c_1, c_2} \,\frac{1}{T}\I(\msg_{[0,T]}, \out_{[0,T]}), \quad \mathrm{(ii)} \; \max_{c_1, c_2} \, \bar{\I}(\msg, \out)$$ subject to one of the constraints (C1) to (C3).

The remainder of the work is organized as follows. Exploiting the link between MI and filtering,
section II addresses the associated filtering problem \cite{snyder1991random} analytically on the level of distributions. We recognize the conditional mean as a piece-wise deterministic Markov process \cite{davis1984piecewise}, and solve the partial differential equation (PDE) given by the hybrid generator \cite{gardiner2009stochastic}. In section III we express the MI and information rate as nonelementary Riemann integrals related to hypergeometric series. This permits further analysis of the constraints (C1) to (C3). 
Capacity-achieving input was reported so far as asymmetric, favoring the OFF state \cite{kabanov1978capacity, davis1980capacity}. In contrast, we find that constraint (C1) either enforces a symmetric allocation of ON and OFF state at optimality for certain $r_1, r_2$ or favors the ON state. ON-favoring local optima are implied by (C3).

\subsection{Previous Work}
\label{sub:previous}
The following expression due to \cite{liptser2001statistics} links the path-wise MI with the filtering problem of observing the input $\msg$ indirectly via the Poisson channel output $\out$. Accordingly,
\begin{equation}
    \I(\msg_{[0,T]}, \out_{[0,T]})  =  \int_0^T \E[\phi(\gain\msg(t))] - \E[\phi(\gain\condmean(t))] \dd t,
    \label{eq:Liptser}
\end{equation} 
where $\phi(z) = z \ln(z)$ and the first conditional moment $\condmean(t) = \E[\msg(t) \mid  \out_{[0,t]}]$ is the optimal causal estimator of the input under a quadratic criterion.

Applying Jensen's inequality to the second integrand term of \eqref{eq:Liptser}
and using the identity $\phi(cz) = c\phi(z) + z\phi(c)$ as well as the stationarity of BMPs, the MI is bounded by
\begin{align*}
    \frac {\I(\msg_{[0,T]}, \out_{[0,T]})} {c_3 T} &\le \frac{c_2 \phi(\state_1) + c_1 \phi(\state_2)}{c_1 + c_2} - \phi\left(\state_1 + \frac{c_1 \Delta \state }{c_1 + c_2}\right)\\
    &=: J(\state_1, \state_2, c_1/(c_1 + c_2)),
\end{align*}
where $\Delta \state := \state_2 - \state_1$ is the dynamic range. Optimizing over the average power constraint $0 < c_1/(c_1 + c_2) \le p_0$ implies
\begin{equation}
    C\le
    c_3J(\state_1, \state_2, \min\{p_0, \bar{p}\}),
    \label{eq:capacity_bound}
\end{equation}
where $C$ can be replaced by $C_T$ and $$\bar{p} = \left\{\exp\left(\frac{\phi(\state_2) - \phi(\state_1)}{\Delta \state} - 1\right) - \state_1\right\}/{\Delta \state}. $$
Kabanov and Davis showed that for $c_1, c_2 \to \infty$ the bound \eqref{eq:capacity_bound} is indeed achieved with the respective asymptotic ratio $c_1/(c_1 + c_2) = \min\{p_0, \bar{p}\}$ \cite{kabanov1978capacity, davis1980capacity}.
The case $\state_1 = 0, \state_2 = 1$ (no leakage) with $p_0  \geq 1/e$ reduces to
$ C = C_T = c_3/e$ with $c_1/(c_1 + c_2) = 1/e. $

We contribute by analyzing how the capacity-achieving BMP input behaves for finite $c_1, c_2$, reflected in the bandwidth-like constraints (C1) to (C3).
For the process class of BMPs, $\condmean(t)$ in \eqref{eq:Liptser} evolves according to the ODE with stochastic jumps \cite{snyder1991random}
\begin{equation}
    \begin{split}
    \dd \condmean(t) =&\left[c_1\cdot\Delta \state - (c_1 + c_2 + c_3 \cdot \Delta \state)(\condmean(t) - \state_1) \right.\\ & \left.+ c_3(\condmean(t) - \state_1)^2 \right] \dd t
    + g(\condmean(t)) \dd \out(t),
\label{eq:filtering}
    \end{split}
\end{equation}
where $\dd \out(t) = 1$ for jump times $t= t_i$, vanishing otherwise; and $g(\condmean(t_i)) = \frac{(\condmean(t_i) - \state_1)(\state_2 - \condmean(t_i))}{\condmean(t_i)}$ is the jump height $\dd \condmean(t_i) = \condmean(t_i+) - \condmean(t_i -)$ at jump times of $\out$. Fig. \ref{fig:cond_mean}b shows a sample trajectory for $\state_1 = 0, \state_2 = 1$. 

\section{The solution of the filtering problem}
In the following, we solve for the distribution of $\condmean(t)$, provided that leakage is absent ($\state_1 = 0$) and the dynamic range is normalized ($\state_2 = 1$). This self-noise limited case allows for an analytic expression of the MI in \ref{sub:MI}.
\begin{figure}
  \centering
    \includegraphics[width=0.49\textwidth]{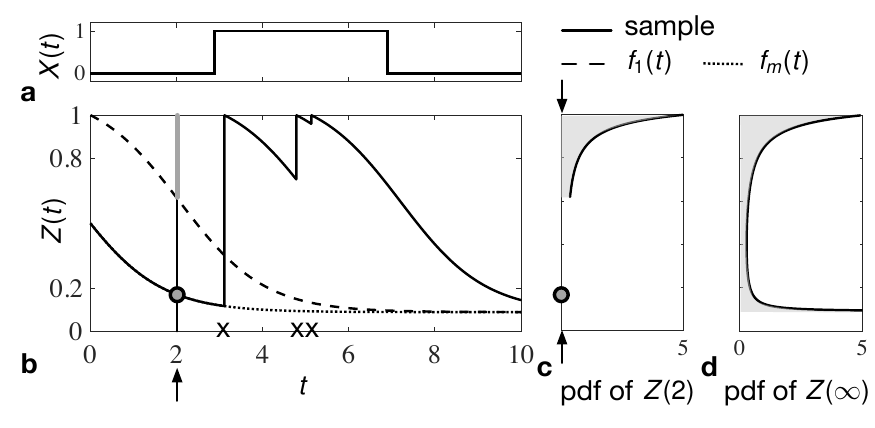}
  \caption{Dynamics of the causal estimator $\condmean(t)$ for $c_1 = c_2 = 0.1, c_3 = 1, \state_1 = 0, \state_2 = 1$. \textbf{a} Sample trajectory of the binary signal input $\msg(t)$. \textbf{b} The solid line is a sample trajectory of $\condmean(t)$. Crosses indicate the jump times $t_1, t_2, t_3$ of $\out$. The dotted line $f_m(t)$ is the common trajectory of $\condmean(t)$ prior to the first jump. The dashed line $f_1(t)$ separates the $(t,z)$-plane. Prior to the first jump the trajectories evolve below, afterwards they evolve above. \textbf{c} The probability distribution of the causal estimator at time $t = 2$ is composed of a Dirac measure with weight $\kappa(2)$ at $f_m(2)$ and a density supported on $(f_1(2), 1]$. \textbf{d} Probability density of the asymptotic causal estimator.}
    \label{fig:cond_mean}
\end{figure}

\subsection{The probability evolution equation}
In absence of leakage the stochastic reset condition in \eqref{eq:filtering} simplifies, because $\out(t)$ increases solely if $\msg(t)$ is in the ON state. Hence, $\condmean(t)$ is reset to $\state_2 = 1$ upon jumps of $\out$.
The joint process $\{\condmean(t), \out(t)\}_{t}$ is then a piece-wise deterministic Markov process \cite{davis1984piecewise} that jumps stochastically from state $\{\condmean(t_i-), \out(t_i-)\}$ to state $\{\condmean(t_i +), \out(t_i +)\} = (1, \out(t_i-) + 1)$. Hence, jump times $t_i$ of $\out$ and of $\condmean$ are identical, and jumps occur with propensity $c_3\condmean(t)$ \cite{snyder1991random}. Since the propensity depends only on the first component, the projection onto $\condmean(t)$ is a piece-wise deterministic Markov process itself. Its probability evolution equation is given by a hybrid generator, composed of the drift (Liouville) and the jump (Poisson) part \cite{gardiner2009stochastic}:
\begin{equation}
   \frac{\partial}{\partial t} p(t,\z) = -\frac{\partial}{\partial \z}\{A(\z)p(t,\z)\} - c_3 \z p(t,\z), \; \omega < \z < 1 , 
   \label{eq:hybrid}
\end{equation}
where
$A(\z) := c_3(c_1 \Delta x/c_3 - \gamma \z + \z^2)$ with $\gamma := (c_1 + c_2 + \Delta x c_3)/c_3$ is the drift dynamics, i.e., the ODE part of \eqref{eq:filtering} with stable equilibrium  $\omega := \frac{\gamma -\rho}{2}$, where $\rho := \sqrt{\gamma^2 - 4c_1\Delta x/c_3}$. The reason why \eqref{eq:hybrid} is lacking an inflow of probability due to jumps is because all inflow enters at $\z = 1$. This is reflected in the proof of the subsequent theorem \ref{theo:prob_evol}. Although expression \eqref{eq:Liptser} has been used and extended \cite{guo2008mutual} in multiple ways to address the capacity problem, to the best of our knowledge, the method introduced here is new to the field. 

\subsection{Transient and asymptotic distribution}
Aiming for the distribution of $\condmean(t)$, we envision the stochastic ensemble of trajectories. All trajectories are initiated in the stationary mean $\statmean :=c_1/(c_1+c_2)$, because $\condmean(0) = \E[\msg(0)]$. They slide down according to the function
$$
f_\iota(t) = \omega + \frac{\rho(\iota-\omega)}{\iota-\omega + e^{\rho c_3 t} (\rho + \omega -\iota)}
$$
with $\iota = m$ (see Fig. \ref{fig:cond_mean}b), which solves the Riccati equation $\frac{\mathrm{d}}{\mathrm{d} t} \condmean(t) =  A(\condmean(t))$.
The solution curve $f_1(t)$ with initial value $\z = 1$ at $t = 0$ separates the $(t,\z)$-plane into two regions summarized in the equivalence that involves the first jump time $t_1$ of $\condmean$:
\begin{equation}
    t < t_1 \Leftrightarrow Z(t) \le f_1(t) \Leftrightarrow \condmean(t) = f_m(t).
    \label{eq:f1(t)}
\end{equation}
Fig. \ref{fig:cond_mean}c visualizes the ensemble of trajectories stopped at a fixed $t$, while Fig. \ref{fig:cond_mean}d visualizes the asymptotic probability distribution. The following theorem fully describes the distribution of $\condmean(t)$ at any time point $t$. Besides preparing the main result theorem \ref{theo:info}, it can be interesting in its own right in the related fields of filtering and control theory.
\begin{theorem}
\label{theo:prob_evol}
The probability measure $\mu_t\colon \mathcal{B}(\omega, 1] \to [0,1], \mu_t(B) = \Prob[Z(t) \in B]$, defined for Borel sets $B \subseteq (\omega,1]$, is a hybrid measure
\begin{equation}
    \mu_t(B) = \kappa(t)\delta_{f_m(t)}(B) + \nu_t(B),
    \label{eq:hybrid_measure}
\end{equation} 
composed of a Dirac measure $\delta_a(B) = \Id_B(a), a \in \R$
at $f_m(t)$ with weight
$$\kappa(t) = e^{-\omega c_3 t}\cdot \left(1 - \rho^{-1}(m - \omega)(1-e^{-\rho c_3 t})\right)$$ and an absolutely continuous measure $\nu_t(\mathrm{d} z) = \pi(\z)\mathrm{d} \z$ supported on $(f_1(t), 1]$ with time-independent density
\begin{equation}
\pi(\z) = \alpha(\z-\omega)^{\beta - \frac{3}{2}}(\omega+\rho - \z)^{-(\beta + \frac{3}{2})},
\label{eq:stat_distr}
\end{equation}
 where
   $$ \alpha := (1-\omega)^{\frac{3}{2} - \beta}(\omega+\rho - 1)^{\beta + \frac{3}{2}}\cdot \frac{c_3 m}{c_2}, \quad \beta := \frac{\gamma}{2\rho}. $$
\end{theorem}
\begin{proof}  We use \eqref{eq:f1(t)} and compute
\begin{align*}
   \Prob[\condmean(t) \in B] &= \Prob[\condmean(t) \in B, t_1 > t] + \Prob[Z(t) \in B, t_1 \le t]\\
    &= \Prob[t_1 > t] \delta_{f_m(t)}(B) + \Prob[\condmean(t) \in B \cap (f_1(t), 1]]
\end{align*}
having the form \eqref{eq:hybrid_measure}. First,
$$ \Prob[t_1 > t] = \exp\left(-\int_0^t f_m(s) \dd s \right) = \kappa(t).$$ Second, by the equivalence
$$\condmean(t) \in B \cap (f_1(t), 1] \Leftrightarrow t - \sup\{t_i \mid t_i \le t \} \in f_1^{-1}(B), t_1 \le t$$
the absolutely continuity of the $t_i$ implies that $\nu_t(B) = \Prob[\condmean(t) \in B \cap (f_1(t), 1]]$ is an absolute continuous measure, i.e., $\nu_t(\mathrm{d}\z) = p(t, \z) \mathrm{d}\z$ with some density $p(t, \z)$ supported on $(f_1(t), 1]$. Its solution is obtained by the method of characteristics \cite{evans1998partial}, initiated at the boundary condition $p_1(t):=p(t,1)$ and propagated through the rewritten linear PDE \eqref{eq:hybrid}
$$ \frac{\partial}{\partial t} p(t,\z)+ A(\z) \frac{\partial}{\partial \z} p(t,\z)= p(t,\z) \left\{ -c_3 \z -  \frac{\mathrm{d}}{\mathrm{d} \z}A(\z)\right\} .$$
It remains to evaluate $p_1(t)$.
We compute
\begin{align*}
    \E[c_3\msg(t)] &= \E \left[ \lim_{h \to 0} \frac{\Prob[\out(t+ h) - \out(t) = 1 \mid \msg_{[0,t]}]}{h}\right]\\
    &= \lim_{h \to 0} \frac{\Prob[\out(t+ h) - \out(t) = 1]}{h}\\
    &= \lim_{h \to 0} \frac{\Prob[\condmean(t+h) \in (f_1(h), 1]]}{h}\\
    &= -f_1'(t) p(t,1)= -A(1) p_1(t) = c_2 p_1(t).
\end{align*}
The absolute continuity of $\nu_t$ was used in the fourth equality. Stationarity of $\msg$ implies that $p_1(t) \equiv c_3m/c_2$ independent of $t$. Plugging this in, $p(t,z) = \pi(z)\Id_{(f_1(t), 1]}(\z)$ is obtained.
\end{proof}

\begin{theorem}
The asymptotic distribution is absolutely continuous, supported on $(\omega, 1]$ with density $\pi(\z)$ as in \eqref{eq:stat_distr}.
\label{theo:prob_asymp} 
\end{theorem}

\begin{proof}
Since $\condmean(t)$ is an ergodic Markov process, the asymptotic equals the stationary distribution. The distribution $\pi(z)$ satisfies the stationarity condition obtained from \eqref{eq:hybrid} by equating the right side to zero.
\end{proof}

\section{The information surface}
\label{sec:info_surface}
We consider the Poisson channel without leakage, unless mentioned otherwise. 
\subsection{Analytic expression}
\label{sub:MI}
The results of the previous section allow us to state the main result of this paper.
\begin{theorem}
\label{theo:info}
Let $\state_1 = 0, \state_2 = 1$, then
\begin{align*}
   \frac 1 T \I(\msg_{[0,T]}, \out_{[0,T]}) =& -\frac {c_3}{T} \int_0^T \phi(f_m(t))\kappa(t) \dd t\\
   &- c_3\int_\omega^1 \phi(\z) \pi(\z) (1-f_1^{-1}(\z)/T) \dd \z,\\
   \bar{\I}(\msg, \out) =& - c_3\int_\omega^1 \phi(\z) \pi(\z) \dd \z.
\end{align*}
\end{theorem}
\begin{proof}
Using theorems \ref{theo:prob_evol},\ref{theo:prob_asymp}, the identity $\phi(cz) = c\phi(z) + z\phi(c)$, and Fubini, evaluate \eqref{eq:Liptser}.
\end{proof}
We use a linear transformation of the integration variable $\z \to (1-\z)(1-\omega)^{-1}$ and the series expansion of the logarithm as well as $z \mapsto (\rho + \omega -z)^{-\beta -3/2}$. Then by uniform convergence of the integrand, we may express
\begin{equation}
    \bar{\I}(\msg, \out) = \sum_{i,j=0}^\infty A_{ij} \left(\frac{1-\omega}{\rho}\right)^j (1-\omega)^{i}\, ,
    \label{eq:series}
\end{equation}
as an absolutely convergent series with coefficients
\begin{align*}A_{ij} &=  \frac{c_3\alpha\rho^{-\beta - \frac{3}{2}} (1- \omega)^{\beta+ \frac{1}{2}}}{\Gamma (\beta + \frac{3}{2})}\times \\
&\frac{i! \Gamma( \beta + \frac{3}{2} +j ) \Gamma( \beta - \frac{1}{2} +j) (\beta - \frac{1}{2} +j  + \omega(i+2))}{j! \Gamma(\beta + \frac{5}{2} +i +j) },
\end{align*}
establishing a link with Appell's hypergeometric $F_3$ series \cite{zbMATH02708199}.
Using \eqref{eq:series}, it can be verified that for $c_1,c_2 \to \infty$ with asymptotic ratio $c_1/(c_1 + c_2) \to p_0$, it holds that $\beta \to 1/2, \omega \to p_0$ and $(1-\omega)/\rho \to 0$. Hence $\bar{\I}(\msg, \out) \to c_3J(0,1,p_0)$, in accordance with paragraph \ref{sub:previous}. Expanding $f_1^{-1}(z)$, the MI $\I(\msg_{[0,T]}, \out_{[0,T]})$ is given in terms of Srivastava's triple hypergeometric series \cite{srivastava1985multiple} by analogous proceeding, implying
$$ \frac 1 T \I(\msg_{[0,T]}, \out_{[0,T]}) \to c_3J(0,1,p_0) $$ for $c_1,c_2 \to \infty$ with asymptotic ratio $c_1/(c_1 + c_2) \to p_0$.

The above linear transformation and subsequent differentiation under the integral sign or, alternatively, summand-wise differentiation of \eqref{eq:series} make (higher-order) partial derivatives with respect to $c_1, c_2$ accessible for the phase plane analysis. 

We introduce the relative rates $\tilde{c}_1 = c_1/c_3$ and $\tilde{c}_2 = c_2/c_3$. The parameters $\gamma, \omega, \rho$ and hence the stationary density $\pi(\z)$ only depend on $\tilde{c}_1$ and $\tilde{c}_2$, and $c_3$ scales the time. The linear time scaling $\tilde{t} = c_3t, \tilde{T} = c_3T$ allows us to write the MI as
$$ \frac 1 T \I(\msg_{[0,T]}, \out_{[0,T]}) = c_3 \cdot \frac{1}{\tilde{T}} \I(\tilde{\msg}_{[0,\tilde{T}]}, \tilde{\out}_{[0,\tilde{T}]}),$$ where $\tilde{\msg}(t)$ and $\tilde{\out}(t)$ are input and output of the correspondent channel with relative rates $\tilde{c}_1,\tilde{c}_2$ and normalized $\tilde{c}_3 = 1$. For convenience we drop the tilde in the following and assume a normalized time scale $c_3 = 1$ for the channel.

\subsection{Constraint (C1) permits ON-favoring input}
 The optimization of $\bar{\I}(\msg, \out) = \bar{\I}(c_1, c_2)$ on the rectangular domain $0 < c_1 \le r_1, 0 < c_2 \le r_2$ in the $(c_1, c_2)$-plane varies qualitatively, depending on the location of the corner $(r_1, r_2)$. As depicted in Fig. \ref{fig:Phase_plane}, we distinguish three regions A, B and C, which are separated by the two nullclines $\partial_1 \bar{\I} = 0$ and $\partial_2 \bar{\I}= 0$. The following analysis relies on the conjecture that both nullclines do not intersect except in the origin.
\begin{figure}
  \centering
    \includegraphics[width=0.4\textwidth]{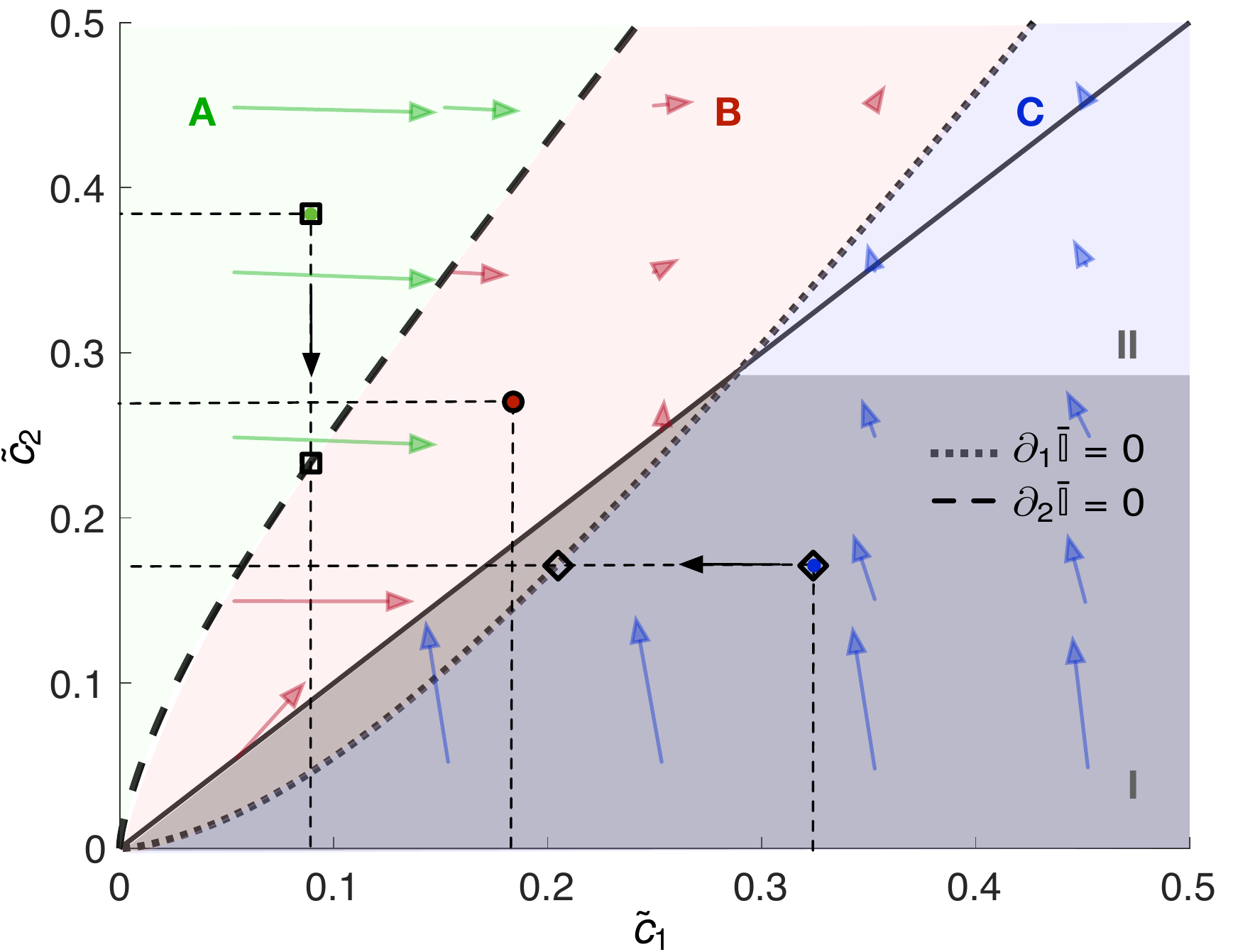}
  \caption{Phase plane analysis of (C1) for $x_1 = 0, x_2 = 1$. Nullclines $\partial_1 \bar{\I} = 0$ and $\partial_2 \bar{\I} = 0$ were evaluated, using the first $100\times 100$ summands of the $F_3$ series in \eqref{eq:series} and deriving summand-wise. Colored arrows indicate the gradient of the information rate, calculated alike. Optimization domains are rectangular. Depending on the location of the domain's upper right corner $(r_1, r_2)$, the optimum is assumed on the nullclines $\partial_1 \bar{\I} = 0$ [$(r_1, r_2)$ in C] and $\partial_2 \bar{\I} = 0$ [$(r_1, r_2)$ in A], respectively, or in the interior of region B [$(r_1, r_2)$ in B]. Regions I and II contain all $(r_1, r_2)$, whose optima favor the ON and OFF state, respectively.}
    \label{fig:Phase_plane}
\end{figure}
Then the three regions are characterized by the signs of the partial derivatives $[\sgn(\partial_1  \bar{\I}), \sgn(\partial_2 \bar{\I})]$, taking on the values $[1, -1], [1,1], [-1,1]$ on A, B, C, respectively. The maximum $(c_1^\ast, c_2^\ast)$ satisfies $c_1^\ast = r_1,c_2^\ast < r_2$ in region A, $c_1^\ast = r_1,c_2^\ast = r_2$ in region B and $c_1^\ast < r_1,c_2^\ast = r_2$ in region C. It is hence located in B or its boundary in all cases. The ratio $\chi := c_1^\ast/ c_2^\ast$ reflects whether the optimal input process $\msg(t)$ favors the ON state ($\chi > 1$) or the OFF state ($\chi < 1$). For large $r_1, r_2$, paragraph \ref{sub:previous} implies that the ratio $\chi$ approaches $1/(e-1)$. The asymmetry of the region B with respect to reflection at the bisection line causes a smaller region I favoring the ON state and a larger region II favoring the OFF state. If we assume a homogeneous constraint (C1) $\max\{c_1, c_2\} \le r_0$, closely related to \cite{shamai1991capacity}, the bisection line transits from region B to region C, upon increasing $r_0$. This constraint thus implies a phase transition from symmetric ($\chi = 1$) to asymmetric ($\chi < 1$) allocation of the ON and the OFF state. An analogous transit behavior for the minimal sojourn time constraint has been reported in \cite{shamai1991capacity}. Fig. \ref{fig:Phase_plane_transient} shows the phase plane associated with the MI, i.e., when $T$ is finite, for comparison with the results obtained for $\bar{\I}(\msg, \out)$.

\subsection{Leakage does not alter the qualitative behavior}
For a channel with leakage ($\state_1 > 0$) the analogous evolution equation $\eqref{eq:hybrid}$ reads
\begin{align*}
    \frac{\partial}{\partial t} p(t,\z) = &-\frac{\partial}{\partial \z}\{A(\z-\state_1) p(t,\z)\} - c_3\z p(t,\z) \\ &+ c_3f_-(\z)p(t,f_-(\z)) \Id(f_-(\z) > \tilde{\omega}), \tilde{\omega} < \z < \state_2,
\end{align*}
where $\tilde{\omega} := \state_1 + \omega$ and $f_-(\z) = \frac{\state_1 \state_2}{\state_1 + \state_2 - \z}$ is the value of $\condmean(t_i-)$ that jumps to $\condmean(t_i +) = \z$.
The delay term $p(t, f_-(\z))$ in the $\z$-component turns it into an equation difficult to solve compared to \eqref{eq:hybrid}. It is yet unclear, whether the asymptotic density $\pi(\z)$ can be solved for, not to mention the time-evolving probability distribution. The first choice of technique for $\pi(\z)$, method of steps \cite{driver2012ordinary}, failed here, because the boundary value $\pi(f^{-1}_-(\tilde{\omega}))$ of the first interval $(\tilde{\omega}, f^{-1}_-(\tilde{\omega})]$ is yet unknown. Furthermore the successive intervals $(f^{-i}_-(\tilde{\omega}), f^{-(i+1)}_-(\tilde{\omega})], i \in \N_0$, where $f^{0}_-(\tilde{\omega}) := \tilde{\omega}, f^{-(i+1)}_- := f^{-1}_- \circ f^{-i}_-$, accumulate, because the sequence $f^{-i}_-(\tilde{\omega})_i$ monotonically approaches $\state_2$.

\begin{figure}
  \centering
    \includegraphics[width=0.48\textwidth]{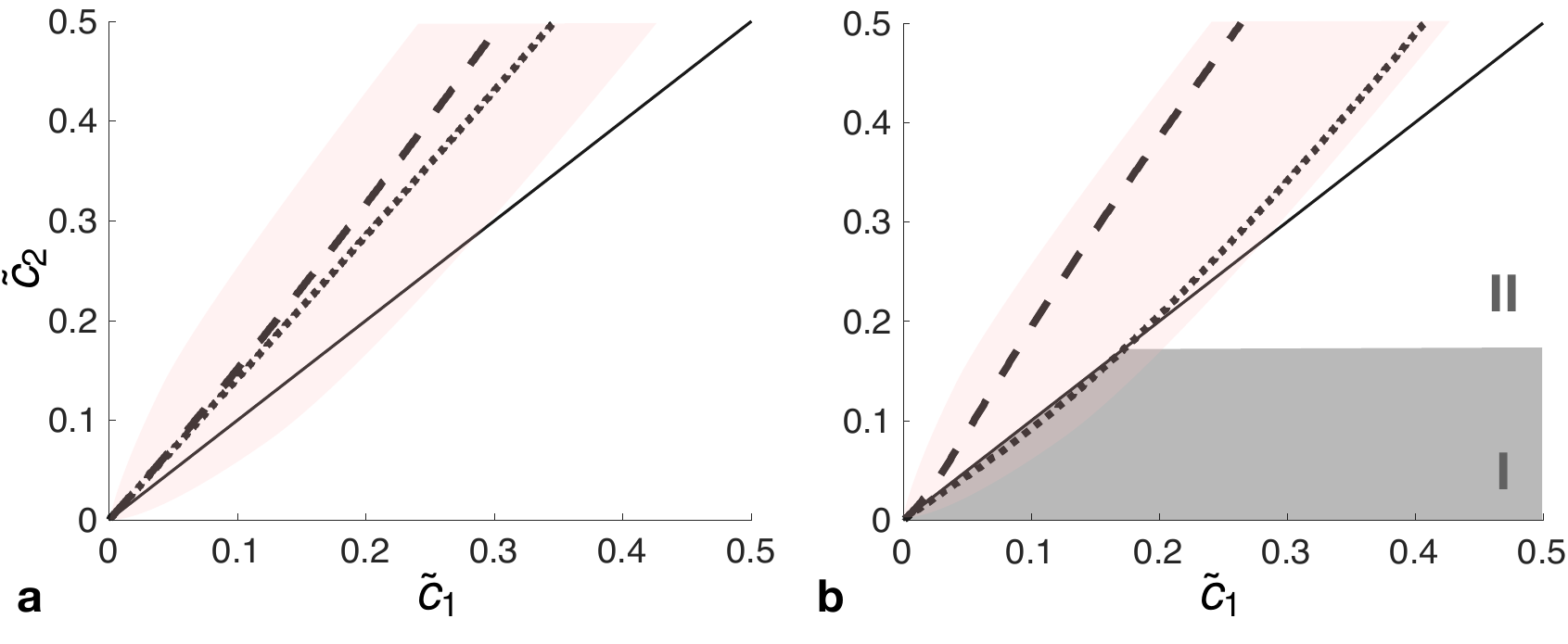}
  \caption{Phase planes for finite $T$. \textbf{a} $T = 1$ \textbf{b} $T = 5$. Dotted and dashed lines indicate $\partial_1 \I = 0$ and $\partial_2 \I = 0$, respectively, and were obtained from numerically evaluating the integral in theorem \ref{theo:info} on a grid with mesh $0.001$. The red shaded area indicates region B of $\bar{\I}$, see Fig. \ref{fig:Phase_plane}, i.e., the case $T \to \infty$.}
    \label{fig:Phase_plane_transient}
\end{figure}

Plots in Fig. \ref{fig:Phase_plane2} show results obtained from stochastic simulation of $\condmean(t)$ with sample size $2\cdot 10^6$ and Monte Carlo evaluation of the mean in \eqref{eq:Liptser}. While the qualitative behavior is not altered, leakage augments the asymptotic ratio from $\chi = 1/(e-1)$ (for $\state_1 = 0$) to $\chi = 1$ (for $\state_1 = 1$) if $\state_2 = 1$ is fixed, according to paragraph \ref{sub:previous}. Increasing leakage thus bends region B towards the bisection line and enlarges the ON-favoring region I. 

\subsection{Constraint \constr{2} enforces asymmetric allocation}
Independent of the input states $\state_1,\state_2$, the autocorrelation time of $\msg$ is $(c_1+c_2)^{-1}$. The optimum of $\bar{\I}$ constrained to the triangular area \constr{2} $c_1+c_2 \leq \rbd$ is achieved on the diagonal boundary. The condition $\partial_{1} \bar{\I} - \partial_{2} \bar{\I}=0$ determines the parametric curve of optimal rates in the $c_1$-$c_2$-plane. Hence it must lie in the region enclosed by the nullclines $\partial_1 \bar{\I} = 0$ and $\partial_2 \bar{\I} = 0$, i.e., in region B of Fig. \ref{fig:Phase_plane}. We observe strictly OFF-favoring, asymmetric, optimal allocations.  For $c_1,c_2 \to 0$ symmetric allocation becomes asymptotically optimal for large minimal autocorrelation times.

\subsection{Constraint \constr{3} allows for local optima}
 The quantity $\E[\E[\sojourn_{\msg(t)}| \msg(t)]]=c_1^{-1} + c_2^{-1} - 2/(c_1 + c_2)$ is the expectation of the average sojourn times of the states $x_1,x_2$ taking into account stationarity of $\msg(t)$.

The constraint \constr{3} allows one rate to be infinite given the other is less than $ \rbd$. Thus its relevance seems to be questionable at first glance. Combination with other constraints that strictly determine the rates to be finite, like \constr{1} and \constr{2}, legitimate \constr{3}. Nevertheless we examined \constr{3} alone.
The constraint defines a family of level sets $\E[\E[\sojourn_{\msg(t)}| \msg(t)]]=\rbd^{-1}$ in the $(c_1, c_2)$-plane, parametrized by $r_0$. The level set for each $r_0$ is a smooth curve, symmetric with respect to the bisection line. Regardless of $r_0$, the maximum of $\bar{\I}$ on each level set is assumed in the OFF-favoring region, i.e., above the bisection line. As $r_0$ decreases, at value $r_0 \approx 0.044$, a local maximum of $\bar{\I}$ appears in the ON-favoring region and perpetuates for smaller $r_0$, resembling a saddle-node bifurcation. Consequently, for average intertransition times greater than approximately $22.7/c_3$ high information throughput can be achieved by an ON-favoring rate pair. Both maxima approach the symmetric allocation as $r_0$ tends to zero.

\begin{figure}
  \centering
    \includegraphics[width=0.48\textwidth]{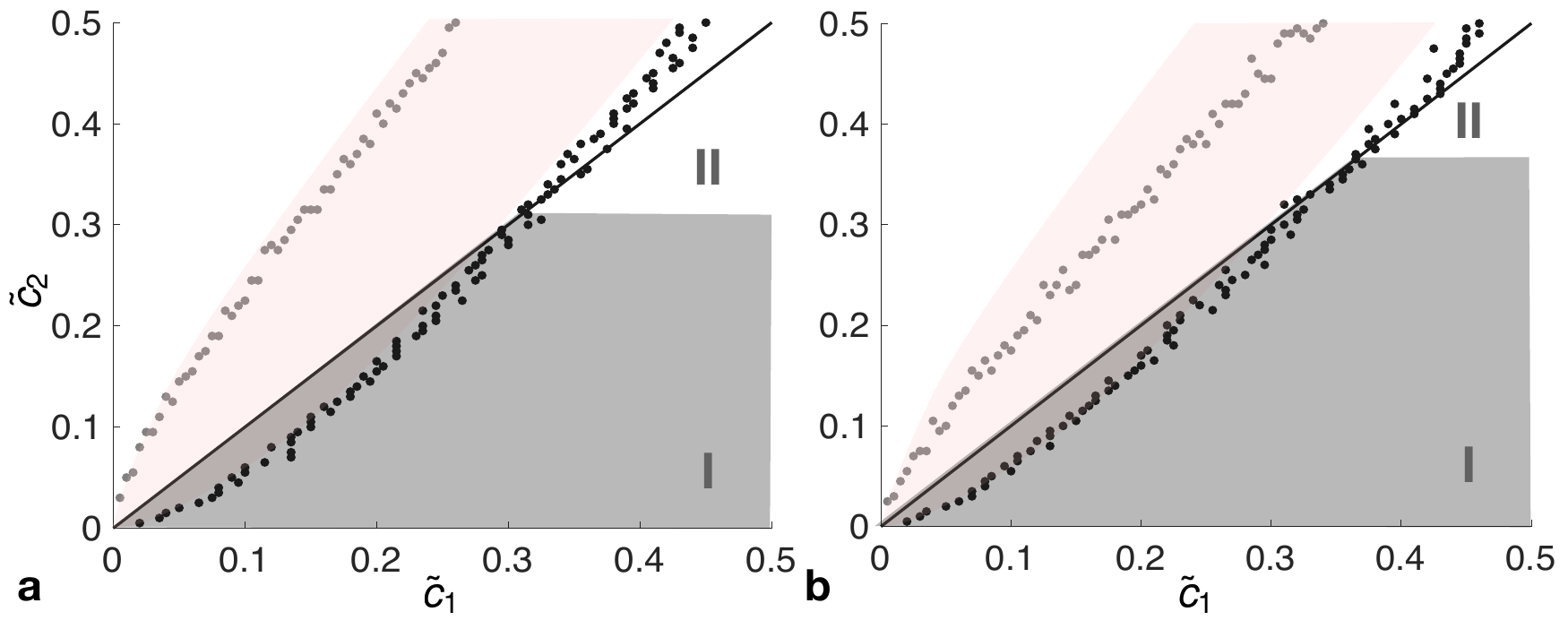}
  \caption{Phase planes with leakage. \textbf{a} $x_1 = 0.01, x_2 = 1$ \textbf{b} $x_1 = 0.1, x_2 = 1$. Black and gray marks indicate the nullclines $\partial_1 \bar{\I} = 0$ and $\partial_2 \bar{\I} = 0$, respectively, and were obtained from Monte Carlo simulations with sample size $2\cdot 10^6$. The red shaded area indicates region B of the case $x_1 = 0$.}
    \label{fig:Phase_plane2}
\end{figure}

\section{Conclusion}
We considered the Poisson channel with BMPs as input class. The method of using the hybrid generator for the evolution of the causal estimator allowed to express its distribution in closed form and represent the MI and information rate by Riemann integrals. The information surface was then analyzed for different constraints on sojourn times. The mathematical derivation heavily relied on the Markov assumption and on neglecting leakage. Advancing mathematical results on leaking, Semi-Markov and multi-state inputs remain open problems. Among general binary inputs, our result yields a lower bound on the capacity.

Optimizing the allocation of the ON and the OFF state under constraint (C1) can be intuitively explained as an interplay between different forces that maximize the efficiency and precision of signal transmission. On the one hand, a force, reducing average sojourn times is predominant in region B of Fig. \ref{fig:Phase_plane}. This force aims at increasing the amount of signals transmitted. On the other hand, forces that increase the sojourn time in the ON and OFF states are predominant in regions A and C, respectively. A larger sojourn time in the ON state increases the likelihood of observing the ON state at the channel output. A larger sojourn time in the OFF state decreases the likelihood of misinterpreting the period between consecutive channel output pulses as an input OFF phase. The phase diagram in Fig. \ref{fig:Phase_plane} explicitly quantifies how the ensemble of forces is balanced. Constraint (C3) admits an analogous driving force towards its local maximum.

Given the analyzed model's rather minimal nature it might not account for biophysical reality. Yet assuming that evolutionary strategies aim at achieving capacity \cite{Tkacik2011Review} the model supplies the hypothesis that system parameters $\tilde{c}_1, \tilde{c}_2$ are located in or near the optimal region B \cite{DeviationsOptimality}. The prediction made by this interpretation is yet to be verified by experimentalists. 

\bibliographystyle{IEEEtran}
\bibliography{bibliography}

\end{document}